\documentclass[letterpaper, 10 pt, conference]{ieeeconf}

\IEEEoverridecommandlockouts                              % This command is only needed if 
                                                   
\usepackage{cite}
\usepackage{amsmath}
\usepackage{amssymb}  % assumes amsmath package installed
\usepackage{accents}
\usepackage{mathtools}
\DeclareMathAlphabet{\mathpzc}{OT1}{pzc}{m}{it}
\usepackage{algorithmic}
\usepackage{algorithm}
\usepackage{array}
\usepackage{caption}
\usepackage{subcaption}
\usepackage{xcolor}
\usepackage{afterpage}
\usepackage{hyperref}
\usepackage{cleveref}
\usepackage{comment}
\usepackage{doi}

\newtheorem{theorem}{Theorem}[section]  
\newtheorem{lemma}[theorem]{Lemma}          
\newtheorem{proposition}[theorem]{Proposition}
\newtheorem{corollary}[theorem]{Corollary}
\newcommand{\norm}[1]{\left\lVert#1\right\rVert}

\begin{document}
% \title{A comparison of sequential quadratic programming and linear-parameter-varying algorithms \\ for real-time model predictive control}
% \title{On sequential quadratic programming \\ and linear-parameter-varying algorithms \\ for real-time model predictive control}
\title{\LARGE \bf Unifying Sequential Quadratic Programming \\ and Linear-Parameter-Varying Algorithms \\ for Real-Time Model Predictive Control}
% \title{A comparison of sequential quadratic programming and linear-parameter-varying }

\author{Kristóf~Floch, Amon~Lahr, Roland~Tóth, and Melanie~N.~Zeilinger % <-this % stops a space
\thanks{K. Floch, A. Lahr and M. Zeilinger are with the Inst.\ for Dynamic Systems and Control, ETH Zürich, Zürich, Switzerland  
{\tt\scriptsize\{kfloch,amlahr,mzeilinger\}@ethz.ch}. R. Tóth is with the Control Systems Group, Eindhoven University of Technology, Eindhoven, The Netherlands and the Systems and Control Lab, HUN-REN Inst.\ for Computer Science and Control, Budapest, Hungary {\tt\scriptsize r.toth@tue.nl}.}%
}

\maketitle

% TODO: add abstract if required
\begin{abstract}
%This paper outlines methods to efficiently solve stochastic \emph{nonlinear model predictive control} (NMPC) problem, speFirst, \emph{sequential quadratic programming} (SQP) is outlined and adapted for the NMPC task, where the original nonlinear system is linearized and a series of QPs are used to solve the original problem. Then, an iterative LPV-MPC scheme is outlined, which similarly to SQP obtains the (sub)optimal control input by the successive solution of computationally efficient QPs. Finally, the similarities and differences between the two methods are investigated.
This paper presents a unified framework that connects \emph{sequential quadratic programming} (SQP) and the iterative \emph{linear-parameter-varying model predictive control} \mbox{(LPV-MPC)} technique. %To derive the LPV-MPC algorithm, we use the exact embedding of the nonlinear dynamics using the Fundamental Theorem of Calculus.
Using the differential formulation of the LPV-MPC, we demonstrate how SQP and LPV-MPC can be unified through a specific choice of scheduling variable and the $2^\mathrm{nd}$ \emph{Fundamental Theorem of Calculus} (FTC) embedding technique and compare their convergence properties. This enables the unification of the zero-order approach of SQP with the \mbox{LPV-MPC} scheduling technique to enhance the computational efficiency of robust and stochastic MPC problems. To demonstrate our findings, we compare the two schemes in a simulation example. Finally, we present real-time feasibility and performance of the zero-order LPV-MPC approach by applying it to \emph{Gaussian process} (GP)-based MPC for autonomous racing with real-world experiments.
\end{abstract}

\section{Introduction}
\label{sec:introduction}

Supported by advances in numerical optimization, \emph{model predictive control} (MPC) has become a key technique for the safety-critical control of dynamical systems due to its ability to handle constraints and its predictive capabilities~\cite{Rawlings_Mayne_Diehl_2017}. %In recent years, a widely researched topic has been the utilization of MPC for nonlinear~\cite{Andersson_Gillis_Horn_Rawlings_Diehl_2019}, and uncertain systems, yielding stochastic~\cite{Feng_Cairano_Quirynen_2020} and robust~\cite{Kohler_Soloperto_Muller_Allgower_2021} \emph{nonlinear} MPC~(NMPC) formulations. However, due to the increased computational complexity of stochastic and robust schemes, real-time implementation of these algorithms on computationally constrained hardware is a prominent area of research.
As most real-world processes exhibit nonlinear behavior, \emph{nonlinear}~MPC~(NMPC) has received increasing attention in recent years~\cite{Andersson_Gillis_Horn_Rawlings_Diehl_2019}. To compute the NMPC input, a \emph{nonlinear program}~(NLP) needs to be solved at every sampling time. For this, two prevalent techniques are \emph{interior point} methods and \emph{sequential quadratic programming}~(SQP)~\cite{Nocedal_Wright_2006}. For real-time NMPC algorithms, particularly SQP methods---iteratively approximating the NLP by a sequence of \emph{quadratic programs}---have gained significant attention due to advances in efficient quadratic programming solvers~\cite{Kouzoupis_Frison_Zanelli_Diehl_2018} and \emph{real-time iteration} schemes~(RTI)~\cite{Diehl_2001}.
% In SQP, the solution of an NLP is approximated by iteratively solving a sequence of \emph{quadratic programs}~(QPs). The primary advantage of using QPs is their convexity, which allows the use of efficient solvers capable of providing real-time solutions~\cite{Kouzoupis_Frison_Zanelli_Diehl_2018}. 
% Furthermore, \emph{real-time iteration} schemes~(RTI)~\cite{Diehl_2001} can enable embedded implementation of the control law. 
%However, the application of SQP algorithms still has limitations: as it linearizes the constraints, it is sensitive to initialization.%, and the number of QP iterations required by the algorithm to converge strongly depends on initial guess of the optimal trajectory.

%for uncertain systems still poses challenges, due to the computational complexity of the uncertainty propagation and the increased number of optimization variables. A promising workaround is the zero-order method \cite{Feng_Cairano_Quirynen_2020, Zanelli_Frey_Messerer_Diehl_2021}, where a tailored Jacobian approximation is used to separate the uncertainty propagation from the optimization, which has also been extended \emph{Gaussian-Process} (GP)-based data-driven models \cite{Lahr_Zanelli_Carron_Zeilinger_2023}.

Closely related to the SQP algorithm,~\cite{Gonzalez_Cisneros_2021} proposes a quasi-linear MPC framework that embeds a nonlinear system into a \emph{linear parameter-varying}~(LPV) form, allowing the NMPC problem to be solved by successive solution of linear MPC problems, corresponding to QPs. %To obtain the LPV representation,~\cite{Gonzalez_Cisneros_2021} exploits the system structure and uses linearization.
In subsequent sections, this method will be referred to as LPV-MPC. Viewing both LPV-MPC and SQP through the lens of inexact Newton-type methods~\cite{Bock_Diehl_Kostina_Schloder}, it can be demonstrated that the two methods have similar convergence properties if the LPV model is obtained by linearization~\cite{Hespe_Werner_2021}. More recently,~\cite{Hoekstra_Cseppento_Beintema_Schoukens_Kollar_Toth_2023} utilizes an automatic FTC-based embedding technique~\cite{Olucha_Koelewijn_Das_Tóth_2025}, to achieve a global LPV representation without approximation. However, the relation of this FTC-based iterative LPV-MPC scheme to SQP has not been thoroughly investigated, and it has yet to be deployed in real hardware experiments.

To further enhance the real-time capabilities of SQP,~\cite{Feng_Cairano_Quirynen_2020,Zanelli_Frey_Messerer_Diehl_2021} propose a zero-order scheme, where a subset of the states of the representation of the system is decoupled from the \emph{optimal control problem} (OCP) through a  tailored Jacobian approximation. This approach was shown to be particularly beneficial for robust~\cite{Zanelli_Frey_Messerer_Diehl_2021}, stochastic~\cite{Feng_Cairano_Quirynen_2020}  and GP-based~\cite{Lahr_Zanelli_Carron_Zeilinger_2023} NMPC \mbox{(GP-MPC)} schemes, where decoupling the uncertainty propagation from the OCP leads to the elimination of the quadratic scaling of the number of optimization variables on the states. Similarly, in the LPV literature, scheduling variables have been utilized to eliminate state variables from the OCP for GP-MPC~\cite{Polcz_Peni_Toth_2023}; however, the connection between these approaches has not yet been explored in the literature.

This paper aims to unify SQP and FTC-based LPV approaches for MPC, yielding the following contributions:
\begin{itemize}
    \item[C1] We introduce a unified solution method for NMPC problems for which the SQP and the FTC-based \mbox{LPV-MPC} schemes are sub-cases. In particular, we show that the FTC embedding approach for LPV-MPC recovers SQP under a specific choice of so-called anchor points.
    \item[C2] We show that the zero-order approximation of the Jacobians can be integrated into the unified scheme and how it can be viewed as using an extended scheduling variable in the LPV-MPC variant.
    \item[C3] We compare computational complexity and the convergence properties of SQP and LPV-MPC in simulation.
    \item[C4] We apply the unified zero-order scheme for the learning-based control of autonomous race cars. Specifically, we implement  both the SQP and LPV-based version of the zero-order GP-MPC algorithm~\cite{Lahr_Zanelli_Carron_Zeilinger_2023} in \textsc{L4acados}~\cite{Lahr_Naf_Wabersich_Frey_Siehl_Carron_Diehl_Zeilinger_2024} to solve the \emph{model predictive contouring control}~(MPCC) problem in simulation and real-world experiments.
\end{itemize}

The remainder of this paper is structured as follows. Sec.~\ref{sec:NMPC} reviews NMPC, introducing both SQP and iterative LPV-MPC as the basis for this paper. Sec.~\ref{sec:unification_sqp_lpv} develops a unified framework that combines the two approaches through a differential formulation and unifies the zero-order method. % for efficient application to uncertain systems.
Then, Sec.~\ref{sec:simulation_study} presents a simulation study, highlighting convergence behavior and computational complexity as well as the applicability of the method for autonomous racing. Finally, Sec.~\ref{sec:experiment} presents the experimental validation.

%The remainder of this paper is organized as follows. First, the general GP-MPC problem is outlined in Sec. \ref{sec:GP_MPC_problem}. Then, we provide a short overview of the zero-order GP-MPC method of \cite{Lahr_Zanelli_Carron_Zeilinger_2023} in Sec. \ref{sec:SQP_GP_MPC}. It is followed by the introduction of the GP-LPV-MPC problem in Sec. \ref{sec:GP_LPV_MPC} and the unification and comparison of the two algorithms in Sec. \ref{sec:unification}. The MPCC implementation of the GP-LPV-MPC problem is described in Sec. \ref{sec:LPV_MPCC}. Finally, simulations and real-world experiments are presented in Sec. \ref{sec:sim_and_exp}.

\section{Nonlinear MPC}
\label{sec:NMPC}
%First, we outline the general NMPC problem, then present the SQP and the iterative LPV-MPC methods as efficient solutions for the original NL predictive problem.
\subsection{General NMPC Problem}
We consider a general \emph{discrete-time}~(DT) \emph{nonlinear}~(NL) system of the form
\begin{equation}
\label{eqn:NL_sys}
    x[k+1]=f(x[k], u[k]),
\end{equation}
where \mbox{$x[k]\in\mathbb{R}^{n_\mathrm{x}}$} is the state vector, \mbox{$u[k]\in\mathbb{R}^{n_\mathrm{u}}$} is the input vector, and $k \in \mathbb{N}$ denotes the discrete time step. The DT state evolution is defined by \mbox{$f: \mathbb{R}^{n_\mathrm{x}}\times \mathbb{R}^{n_\mathrm{u}}\rightarrow \mathbb{R}^{n_\mathrm{x}}$}.

For simplicity, as a control objective, we focus on stabilizing the system around an equilibrium point, assumed w.l.o.g.\ to be at the origin. However, we note that the extension to tracking tasks is straightforward, see~\cite{Rawlings_Mayne_Diehl_2017}.  Furthermore, we prescribe state and input constraints as
%\begin{equation}
%\label{eqn:nonlinear_constraint}
\mbox{$h(x[k],u[k])\leq 0,\; h_N(x[k])\leq 0$},
%\end{equation}
where \mbox{$h: \mathbb{R}^{n_\mathrm{x}} \times \mathbb{R}^{n_\mathrm{u}} \rightarrow\mathbb{R}^{n_\mathrm{h}}$} and \mbox{$h_N: \mathbb{R}^{n_\mathrm{x}} \rightarrow\mathbb{R}^{n_\mathrm{h_N}}$}.

Generally, in NMPC, given a current state measurement $x[k]$ at time step $k$, an NLP is solved to obtain a sequence of optimal state and input values over a finite prediction horizon. The NLP can be formulated as follows:
\begin{subequations}
\label{eqn:NLP}
\begin{align}
\min_{X,U} \quad&\sum_{i=0}^{N-1}l_i(x_i,u_i)+l_N(x_N)\label{eqn:NL_MPC_cost}\\
\textrm{s.t.}&\quad \forall i\in\mathbb{I}_0^{N-1}\\
&\quad x_{i+1}=f(x_i,u_i),\label{eqn:state_progagation}\\
             &\quad h(x_i,u_i)\leq0, \label{eqn:nlp_constraints}\\
             &\quad h_N(x_N)\leq0, \label{eqn:nlp_terminal_constraints}\\
             &\quad x_0=x[k],
\end{align}
\end{subequations}
%$z\coloneqq \mathrm{vec}(X,U) \in\mathbb{R}^{2N+1}$ is the vectorized state and input vector, $\xi : \mathbb{R}^{2N+1}\rightarrow \mathbb{R}$ encodes the cost function, $\eta: \mathbb{R}^{2N+1}\rightarrow\mathbb{R}^{N}$ encodes the system dynamics, i.e., $[\Tilde{\eta}(z)]_{i}\coloneqq x_{i+1}-f(z_i),\, i\in\mathbb{I}_0^{N-1}$ and ${\zeta}(z): \mathbb{R}^{2N+1}\rightarrow\mathbb{R}^{n_\mathrm{c}}$ contains the input constraints \eqref{eqn:U_const} and state constraints \eqref{eqn:X_const}-\eqref{eqn:Xf_const}.
where \mbox{$l_i:\mathbb{R}^{n_\mathrm{x}}\times \mathbb{R}^{n_\mathrm{u}}\rightarrow \mathbb{R}$} denotes the stage cost, \mbox{$l_N:\mathbb{R}^{n_\mathrm{x}}\rightarrow \mathbb{R}$} is the terminal cost, $N$ is the control horizon, \mbox{$U=[u_0^\top\dots u_{N-1}^\top]^\top$} and \mbox{$X=[x_0^\top\dots x_N^\top]^\top$} are the optimization variables and \mbox{$\mathbb{I}_{\tau_1}^{\tau_2}=\{i\in\mathbb{Z}\;|\;\tau_1\leq i \leq \tau_2 \}$}. For simplicity, in this paper we consider quadratic stage and terminal costs defined as
    \begin{equation}
    \label{eqn:quadratic_stage_and_terminal_cost}
        l_i(x_i,u_i)  \doteq \norm{x_i}_Q^2 + \norm{u_i}_R^2, \quad
        l_N(x_N) \doteq \norm{x_N}_W^2,
    \end{equation}
where $Q\succeq0$, $R\succ 0$, and $W\succeq0$ are the positive \mbox{(semi-)}definite weighting matrices and \mbox{$\norm{a}_B^2\doteq a^\top B a$}. %Note that in this formulation, the terminal constraint is omitted for simplicity.
In the following, we assume that all functions in the NLP \eqref{eqn:NLP} are at least twice continuously differentiable. 

\subsection{SQP Solution}
\label{sec:SQP}
The main idea of the SQP-based solution is that, given an initial guess of \mbox{$\hat{X}=[\hat{x}_0^\top\dots\hat{x}_N^\top]^\top,\;\hat{U}=[\hat{u}_0^\top\dots\hat{u}_{N-1}^\top]^\top$}, the original problem \eqref{eqn:NLP} is approximated by a single QP, which provides a reliable approximation in a local neighborhood around the linearization points, $\hat{X}, \hat{U}$. By defining the optimization variables as \mbox{$\Delta x_i=x_i-\hat{x}_i,\; \Delta u_i=u_i-\hat{u}_i$}, the QP yields $\Delta X^\star$,$\Delta U^\star$. Then, the current approximation is updated as \mbox{$\hat{X}\leftarrow \hat{X}+\Delta X^\star$}, \mbox{$\hat{U}\leftarrow \hat{U}+\Delta U^\star$} to obtain a sequence of solutions that is, under certain conditions, proven to converge to a Karush-Kuhn-Tucker (KKT) point of \eqref{eqn:NLP}, denoted by $X^\star, U^\star$, cf.~\cite[Thm.~1]{Boggs_Tolle_1995}. %, according to Lemma~\ref{lem:sqp-conv}.%\cite[Chap. 18]{Nocedal_Wright_2006}.
The quadratic subproblem solved at each SQP iteration can be defined as
\begin{subequations}
\label{eqn:QP_of_SQP}
    \begin{align}
    \min_{\Delta X, \Delta U} &\quad \sum_{i=0}^{N-1}
    \frac{1}{2}
     \begin{bmatrix}\Delta x_i \\ \Delta u_i \end{bmatrix}^\top \mathpzc{M}_i \begin{bmatrix}
        \Delta x_i\\
        \Delta u_i
    \end{bmatrix} + m_i^\top (
        \hat{x}_i,\hat{u}_i)\begin{bmatrix}
        \Delta x_i\\ \Delta u_i
    \end{bmatrix}\nonumber \\
    & + \frac{1}{2}\Delta x_N^\top \mathpzc{M}_N \Delta x_N + m_N^\top({\hat{x}_N}) \Delta x_N \\
        %\min_{\Delta z} &\quad \frac{1}{2} \Delta z^{\top} M \Delta z+\frac{\partial \xi}{\partial z}(\hat{z}) \Delta z, \\
        %\textrm{s.t.} &\quad \eta(\hat{z})+ \frac{\partial \eta}{\partial z}(\hat{z})\Delta z=0,\\
        %&\quad \zeta(\hat{z})+ \frac{\partial \zeta}{\partial z}(\hat{z})\Delta z\leq0,
        \textrm{s.t.} & \quad \forall i\in\mathbb{I}_0^{N-1}\\
        &\quad \Delta x_{i+1}=A_i\Delta x_i+B_i\Delta u_i + \Delta\mathpzc{W}_i^\mathrm{x},\;  \\
                   & \quad 0\geq H^\mathrm{x}_i\Delta x_i+H_i^\mathrm{u}\Delta u_i+ \Delta\mathpzc{W}_i^\mathrm{h},\\
                   & \quad 0\geq H^\mathrm{x}_N\Delta  x_N+\Delta \mathpzc{W}_N^\mathrm{h},\\
                  & \quad \Delta x_0=0.
    \end{align}
\end{subequations}
%In \eqref{eqn:QP_of_SQP}, the terms are defined as follows.
In the cost of \eqref{eqn:QP_of_SQP}, $\mathpzc{M}_i$ is the chosen approximation of the Hessian of the Lagrangian and $m_i$ is the Jacobian of the original cost at step $i$, which for a quadratic cost evaluates to $m_i = {M}_i[x_i^\top\; u_i^\top]^\top$. The derivation of $\mathpzc{M}_i$ is included in Appendix~\ref{app:hessian}.
The state and input matrices for the linearized dynamics and constraints are the Jacobians of \eqref{eqn:state_progagation} and \eqref{eqn:nlp_constraints} \eqref{eqn:nlp_terminal_constraints}, respectively, evaluated at $\hat{x}, \hat{u}$, i.e.,
    \begin{align}
    \label{eqn:SQP_AB_jacobian}
        A_i = \left. \frac{\partial f} {\partial x} \right | _ {\substack{\hat{x}_i \\ \hat{u}_i}}, \, 
        B_i = \left. \frac{\partial f} {\partial u} \right | _ {\substack{\hat{x}_i \\ \hat{u}_i}}, \,
        H_i^\mathrm{x}= \left. \frac{\partial h} {\partial x} \right | _ {\substack{\hat{x}_i \\ \hat{u}_i}}, \,
        H_i^\mathrm{u}= \left. \frac{\partial h} {\partial u} \right | _ {\substack{\hat{x}_i \\ \hat{u}_i}},
    \end{align}
%\end{subequations}
and the residual terms are 
\begin{subequations}
\begin{align}
    \Delta\mathpzc{W}_i^\mathrm{x}&={f}(\hat{x}_i, \hat{u}_i)-\hat{x}_{i+1},\\\Delta\mathpzc{W}^\mathrm{h}_i&={h}(\hat{x}_i, \hat{u}_i),\; i\in\mathbb{I}_0^{N-1},\;
    \Delta\mathpzc{W}^\mathrm{h}_N={h}(\hat{x}_N).
\end{align}
\end{subequations}
Using \eqref{eqn:QP_of_SQP}, the standard SQP algorithm is summarized in Alg.~\ref{alg:SQP-MPC}. 
In particular, if the solution of the quadratic subproblem yields a vanishing step 
\mbox{$\Delta X = 0$}, \mbox{$\Delta U =0$}, then, according to Lemma~\ref{lem:sqp-conv}, the current iterate satisfies the KKT conditions 
of the original nonlinear program.

\begin{algorithm}
    \caption{SQP-based MPC.}
    \label{alg:SQP-MPC}
    \begin{algorithmic}[1]
        \STATE \textbf{input} $x[k]$ (measured state at time~$k$), ${X}^\star, {U}^\star$ (optimal state/input sequence at time~$k-1$)
        \STATE \textbf{initialize} $\hat{X}=X^\star, \hat{U}=U^\star$
        \REPEAT
            \STATE \textbf{set} $A_i, B_i, i \in \mathbb{I}_{i=0}^{N-1}$ via \eqref{eqn:SQP_AB_jacobian}
            \STATE \textbf{solve} \eqref{eqn:QP_of_SQP} to obtain $\Delta X^\star, \Delta U^\star$
            \STATE \textbf{update} $\hat{X} := \hat{X} + \Delta X^\star, \; \hat{U} := \hat{U} + \Delta U^\star$
         \UNTIL{convergence criterion is reached
        \STATE \textbf{set} $X^\star=\hat{X},\; U^\star = \hat{U}$}
        \STATE \textbf{apply} $u[k]=[U^\star]_0$
    \end{algorithmic}
\end{algorithm}

\begin{lemma}[Optimality of SQP \protect{\cite[Chap.~18]{Nocedal_Wright_2006}}]
\label{lem:sqp-conv}
Consider the  NMPC problem \eqref{eqn:NLP}, solved by SQP using Alg.~\ref{alg:SQP-MPC}. Suppose that standard SQP assumptions hold~\cite[Sec.~3.1]{Boggs_Tolle_1995} and the algorithm has converged, i.e., \mbox{$\Delta X^\star=0,\;\Delta U^\star=0$}. Then, \({X}^\star, {U}^\star\) together with the corresponding Lagrange multipliers
satisfy the KKT first‑order optimality conditions of the original NLP~\eqref{eqn:NLP},
in particular, \({X}^\star, {U}^\star\) is a first‑order stationary point of the NLP, i.e., a locally optimal solution.
\end{lemma}

During optimization, KKT residuals of the original NLP~\eqref{eqn:NLP} are commonly used as a convergence criterion for SQP implementations~\cite[Eq.~(12.34)]{Nocedal_Wright_2006}.

\subsection{Iterative LPV-MPC}
\label{sec:LPV-MPC}
An alternative approach to solve \eqref{eqn:NLP} is to utilize an LPV embedding in an iterative LPV-MPC scheme~\cite{Gonzalez_Cisneros_2021}.
%, originally proposed in 
To discuss this approach, first, the LPV embedding of NL systems is outlined using the FTC, based on~\cite{Olucha_Koelewijn_Das_Tóth_2025}. Then, the iterative solution of the LPV-MPC problem is presented.

\subsubsection{LPV systems}
\label{sec:LPV_systems_embedding}
%Consider an LPV system in the form
%\begin{equation}
%\label{eqn:LPV_sys}
%    x[k+1]=A(\rho[k])x[k]+B(\rho[k])u[k],
%\end{equation}
%where \mbox{$A: \mathbb{R}^{n_\rho}\rightarrow \mathbb{R}^{n_\mathrm{x}\times n_\mathrm{x}}$}, \mbox{$B: \mathbb{R}^{n_\rho}\rightarrow \mathbb{R}^{n_\mathrm{x}\times n_\mathrm{u}}$} are the parameter-dependent state and input matrix, respectively, and $\rho\in\mathbb{R}^{n_\rho}$ is the \emph{scheduling variable}. By defining a scheduling map \mbox{$\rho[k]=\Phi(x[k],u[k])$}, we highlight that $\rho$ can be obtained from the state and input trajectories. We aim to embed the original nonlinear dynamics \eqref{eqn:NL_sys} into LPV form \eqref{eqn:LPV_sys}, i.e. \mbox{$f(x,u)=A(\rho)x+B(\rho)u$}. In general, any nonlinear system can be represented as a linear system, with state- and input-dependent matrices $A, B$. By neglecting the dependence of the matrix variables on the states, different LPV representations can be obtained~\cite{Abbas_Toth_2014}. However, many of them scale poorly or require expert knowledge. In this paper, we focus on the FTC-based formulation proposed by~\cite{Olucha_Koelewijn_Das_Tóth_2025}, whose benefit is that it provides the embedding of the exact nonlinear dynamics into an LPV representation automatically, without requiring any prior information about the nonlinear dynamics. 

% In general, 
A wide class of nonlinear functions can be represented as
% \begin{align*}
\mbox{$f(x, u) = A(\Phi(x,u)) x + B(\Phi(x,u)) u + V$},    
% \end{align*}
where
the matrices
\mbox{$A: \mathbb{R}^{n_\rho}\rightarrow \mathbb{R}^{n_\mathrm{x}\times n_\mathrm{x}}$}, \mbox{$B: \mathbb{R}^{n_\rho}\rightarrow \mathbb{R}^{n_\mathrm{x}\times n_\mathrm{u}}$}
% are matrices 
depend on the states and inputs via the so-called scheduling map \mbox{$\Phi: \mathbb{R}^{n_\mathrm{x}}\times \mathbb{R}^{n_\mathrm{u}} \rightarrow \mathbb{R}^{n_\rho}$}, and \mbox{$V \in \mathbb{R}^{n_\mathrm{x}}
$} is either a constant offset vector or it can also be dependent on $\Phi$. By defining the \emph{scheduling variable} as $\rho[k] \doteq \Phi(x[k],u[k])$ the LPV representation~\cite{Abbas_Toth_2014} of \eqref{eqn:NL_sys} is
\begin{equation}
\label{eqn:LPV_sys}
    x[k+1] = A(\rho[k])x[k] + B(\rho[k])u[k] + V[k],
\end{equation}
where the dependence of $\rho$ on $x$ and $u$ is intentionally neglected to obtain an embedding of the NL system, enabling convex analysis and synthesis, or is treated as a known or uncertain sequence in predictive control. There exists a wide variety of methods to accomplish the factorization of the nonlinearity to obtain the LPV representation \eqref{eqn:LPV_sys}. Many of these methods are only applicable to specific model structures, are computationally demanding, or require expert decisions in the modeling process, cf.~\cite{Olucha_Koelewijn_Das_Tóth_2025}. To establish the connection between SQP and LPV methods, we focus on the FTC-based formulation proposed by~\cite{Olucha_Koelewijn_Das_Tóth_2025}, which automatically embeds the exact nonlinear dynamics into an LPV representation without requiring manual design choices.

\allowdisplaybreaks

Given a continuously differentiable function \mbox{$a: \mathbb{R}^n\rightarrow\mathbb{R}^m$}, % $t \mapsto a(t)$,
the FTC states that, for $\eta, \tilde{\eta}\in\mathbb{R}^{n}$,
\begin{equation}
\label{eqn:FTC}
    a(\eta)-a(\tilde{\eta})=\left(\int_0^{1}\left.\frac{\mathrm{d}a}{\mathrm{d}\eta}\right|_{\tilde{\eta} + \lambda (\eta-\tilde{\eta})}  \mathrm{d}\lambda\right)(\eta-\tilde{\eta}),
\end{equation}
where $\left.\frac{\mathrm{d}a}{\mathrm{d}\eta}\right|_{\lambda \eta}$ is the Jacobian of $a$ evaluated at $\lambda\eta$. Since $f$ is differentiable, by choosing \mbox{$\eta_i = [x_i^\top\; u_i^\top]^\top$} and \mbox{$\tilde{\eta} = [\tilde{x}_i^\top\; \tilde{u}_i^\top]^\top$}, we obtain
\begin{multline}
\label{eqn:FTC_to_LPV}
    f(x_i,u_i)=  \underbrace{\int_0^1 \left. \frac{\partial f}{\partial x}\right|_{\substack{\tilde{x}_i+\lambda (x_i-\tilde{x}_i) \\ \tilde{u}_i+ \lambda (u_i-\tilde{u}_i)}}\mathrm{d}\lambda}_{\bar{A}(x_i,u_i)}(x_i-\tilde{x}_i)\\ + \underbrace{\int_0^1\left.\frac{\partial f}{\partial u}\right|_{\substack{\tilde{x}_i + \lambda (x_i-\tilde{x}_i)\\\tilde{u}_i + \lambda (u_i-\tilde{u}_i)}}\mathrm{d}\lambda}_{\bar{B}(x_i,u_i)}(u_i-\tilde{u}_i) + V_i,
\end{multline}
where \mbox{$V_i=f(\tilde{x}_i,\tilde{u}_i)$} is a term dependent on the anchor points $\tilde{x}_i, \tilde{u}_i$. By defining the scheduling map as \mbox{$\Phi(x_i,u_i)\doteq[x_i^\top\; u_i^\top]^\top$}, the scheduling-dependent state and input matrices of the LPV form become
%\begin{equation*}
    \mbox{$A(\rho_i) \doteq \bar{A}(x_i,u_i), \quad B(\rho_i) \doteq   \bar{B}(x_i,u_i)$}.
%\end{equation*}
%As a result, we obtain an LPV representation, plus an additional affine term, i.e.,
%\begin{equation}
%    \label{eqn:LPV_sys_plus_affine}
%    x[k+1]=A(\rho[k])(x[k]-\tilde{x})+B(\rho[k])(u[k]-\tilde{u}) + V(\tilde{x},\tilde{u}).
%\end{equation}
Note that in most LPV applications, the anchor points $\tilde{x}_i, \tilde{u}_i$ are considered constant-zero with $0=f(0,0)$, which naturally gives the LPV form. In contrast, this paper also investigates non-zero and varying anchor points along the prediction horizon.
It is also important to highlight that, for the LPV conversion, the calculation of the integrals of the Jacobians is required. While the Jacobians can be easily computed using symbolic computation packages or algorithmic differentiation, the analytical expression of the integrals is a difficult task. However, as only the values of the $A$ and $B$ matrices are interesting in an MPC formulation at a given $\rho[k]$, we can rely on numerical integration methods, which can also be parallelized \cite[Sec.~V.F]{Hoekstra_Cseppento_Beintema_Schoukens_Kollar_Toth_2023} for efficiency. Note that it is straightforward to apply \eqref{eqn:FTC} to the nonlinear inequality constraints~\eqref{eqn:nlp_constraints}-\eqref{eqn:nlp_terminal_constraints}, yielding a similar parameter-varying formulation with $H^\mathrm{x}, H^\mathrm{u}$ defined in \Cref{TABLE:SQP_vs_LPV}.

\subsubsection{LPV-MPC}
Using the LPV model obtained in Sec.~\ref{sec:LPV_systems_embedding}, we can employ the method outlined in \cite{Gonzalez_Cisneros_2021} to solve the NMPC problem efficiently using an iterative procedure. The key idea of the LPV-MPC approach is that at any given time step $k$, for a fixed scheduling sequence \mbox{${P}=[\rho_{0}^\top\dots \rho_{N-1}^\top]^\top$}, Eq.~\eqref{eqn:LPV_sys} reduces to an affine system, for which the MPC problem can be efficiently solved via the following QP:
\begin{subequations}
\label{eqn:q_linear_MPC}
\begin{align}
    \min_{X,U}& \;\, \sum_{i=0}^{N-1}(\norm{x_i}_Q^2 + \norm{u_i}_R^2)+\norm{x_N}_W^2,\\
    \textrm{s.t.} &\;\, \forall i\in\mathbb{I}_0^{N-1} \\
                     & \;\, x_{i+1}=A(\rho_i)(x_i-\tilde{x}_i)+B(\rho_i)(u_i-\tilde{u}_i) + V^\mathrm{x}_i,\; \label{eqn:q_lin_state_const} \\
                   & \;\, 0\geq H^\mathrm{x}(\rho_i)(x_i-\tilde{x}_i)+H^\mathrm{u}(\rho_i)(u_i-\tilde{u}_i)+ V^\mathrm{h}_i, \label{eqn:q_linear_const1}\\
                   & \;\, 0\geq H^\mathrm{x}(\rho_N)(x_N-\tilde{x}_N)+ V^\mathrm{h}_N, \label{eqn:q_linear_const2}\\
                  & \;\, x_0=x[k].
\end{align}
\end{subequations}
As ${P}$ is fixed in \eqref{eqn:q_linear_MPC}, state propagation reduces to LTV dynamics. Furthermore, \eqref{eqn:q_linear_const1} is the LPV factorization of the constraints~\eqref{eqn:nlp_constraints}. As a result, \eqref{eqn:q_linear_MPC} corresponds to a quadratic subproblem that can be efficiently solved by standard QP solvers \cite{Kouzoupis_Frison_Zanelli_Diehl_2018}. Solving \eqref{eqn:q_linear_MPC} yields an optimal input sequence $U^\star$, which can be used to forward-simulate the NL model~\eqref{eqn:NL_sys} to obtain the scheduling sequence at time step $k+1$, based on which a new quadratic subproblem can be formulated and solved. By executing this iteration until the input trajectory has converged, the solution of the quadratic subproblem converges to a suboptimal solution of the NMPC, assuming that the LPV approximation is sufficiently accurate, according to Lemma~\ref{lem:lpvmpc-suboptimal}.

\begin{lemma}[Suboptimality of LPV--MPC \protect{\cite[Thm.~3]{Hespe_Werner_2021}}]\label{lem:lpvmpc-suboptimal}
Consider the LPV--MPC problem \eqref{eqn:q_linear_MPC}, with a scheduling trajectory 
$\rho_i=\Phi(\hat{x}_i, \hat{u}_i)$ forming the QP approximation according to \eqref{eqn:q_linear_MPC} and  Alg.~\ref{alg:LPV-MPC}.  
Suppose the algorithm has converged, %i.e., $\Delta X =0,\;\Delta U = 0$,therefore
i.e., the solution $(X^\star,U^\star)$ of the QP coincides with the current iterate.  
Then $(X^\star,U^\star)$ is feasible and in general a suboptimal solution for the original NMPC \eqref{eqn:NLP}.
\end{lemma}

The iterative LPV-MPC algorithm is outlined in Alg.~\ref{alg:LPV-MPC}.

\begin{algorithm}
    \caption{Iterative LPV-MPC.}
    \label{alg:LPV-MPC}
    \begin{algorithmic}[1]
        \STATE \textbf{input} $x[k]$ (measured state at time $k$), ${U}^\star$ (optimal input sequence at $k-1$)
        \STATE \textbf{initialize}\footnotemark[1] $[{P}]_i=\Phi(x[k], [{U}^\star]_{i+1}),\; i\in\mathbb{I}_{i=0}^{N-1}$
        \REPEAT
            \STATE \textbf{set} $A(\rho_i), B(\rho_i), i \in \mathbb{I}_{i=0}^{N-1}$ via \eqref{eqn:FTC_to_LPV}
            \STATE \textbf{solve} \eqref{eqn:q_linear_MPC} to obtain ${U}$
            \STATE \textbf{simulate} \eqref{eqn:NL_sys} with $U^\star\leftarrow {U}$ and $x[k]$ to obtain ${X}$ \label{alg:line:sim}
            \STATE \textbf{set} $[{P}]_i=\rho_i=\Phi([{X}]_i,[{U}]_i),\; i\in\mathbb{I}_{i=0}^{N-1}$
         \UNTIL{${P}$ has converged}
        \STATE \textbf{apply} $u[k]=[U^\star]_0$
    \end{algorithmic}
\end{algorithm}
\footnotetext[1]{{For $i\geq2$, $X^\star$ can also be used to initialize $P$.}}
As convergence criterion, most LPV-MPC approaches monitor the convergence of the scheduling variables, i.e., if
\begin{equation}
\label{eqn:lpv_convergence_criteria}
    \norm{{P}-\hat{{P}}}_\infty \leq \epsilon_\mathrm{LPV},
\end{equation}
where $\hat{{P}}$ denotes the previous scheduling sequence.

\section{Unifying SQP and LPV-MPC}
\label{sec:unification_sqp_lpv}
\subsection{Equivalence Condition}
\label{sec:unified_nominal}
To show how the SQP and LPV-MPC approaches are related, we reformulate the QP of the LPV-MPC problem into a differential form akin to SQP, i.e., we use $\Delta X$ and $\Delta U$ as optimization variables, similarly to \cite{Karachalios_Abbas_2024}. First, along the trajectory $\hat{X}, \hat{U}$, the state-evolution constraint \eqref{eqn:q_lin_state_const} can be expressed as
\begin{equation}
x_{i+1} = A(\hat{\rho}_i)x_i + B(\hat{\rho}_i)u_i + \mathpzc{W}^\mathrm{x}_i, 
\end{equation}
with \mbox{$\mathpzc{W}_i=V^\mathrm{x}_i - A(\hat{\rho}_i)\tilde{x}_i - B(\hat{\rho}_i)\tilde{u}_i$}, which is completely determined by the stage-wise anchor points $\tilde{X}, \tilde{U}$ and the trajectory $\hat{X}, \hat{U}$. By %adding and subtracting both \mbox{$\hat{x}_{i+1}$} and \mbox{$A(\hat{\rho}_i)\hat{x}_i +B(\hat{\rho}_i)\hat{u}_i$} from and to the left and right side, respectively, and 
defining \mbox{$\Delta x_i\coloneqq x_{i} - \hat{x}_{i}$, $\Delta u_i\coloneqq u_{i} - \hat{u}_{i}$}, we arrive at 
\begin{align}
    \hat{x}_{i+1} + \Delta x_{i+1} &= A(\hat{\rho}_i)\Delta{x}_i + B(\hat{\rho}_i)\Delta u_i \\&\phantom{=} + A(\hat{\rho}_i)\hat{x}_i + B(\hat{\rho}_i)\hat{u}_i  + \mathpzc{W}_i^\mathrm{x}.\nonumber
\end{align}
Finally, by rearranging the terms, the state propagation in differential form can be expressed as 
\begin{equation}
\label{eqn:state_prop_diff}
        \Delta x_{i+1} = A(\hat{\rho}_i)\Delta{x}_i  + B(\hat{\rho}_i)\Delta u_i  + \Delta{\mathpzc{W}}^\mathrm{x}_i,
\end{equation}
where the residual term is
\begin{equation}
\Delta{\mathpzc{W}}_i^\mathrm{x} = -\hat{x}_{i+1} + \underbrace{A(\hat{\rho}_i) (\hat{x}_i-\tilde{x}_i) + B(\hat{\rho}_i) (\hat{u}_i-\tilde{u}_i) + V_i^\mathrm{x}}_{f(\hat{x}_i,\hat{u}_i)},
\end{equation}
according to the FTC-based factorization \eqref{eqn:FTC_to_LPV}.

Finally, we outline a general unified notation that can be employed for both the SQP and the LPV-MPC methods by introducing the following standard quadratic form: 
%\begin{subequations}
%\begin{align}
%    \min_{\Delta z} & \quad \frac{1}{2}\Delta z^\top \mathpzc{M}  \Delta z + (\mathpzc{M}z)^\top \Delta z\\
%          \textrm{s.t.}          & \quad \forall i=0,\dots,N-1 \\
%                                 & \quad \Delta x_{i+1} = A_i\Delta x_i + B_i\Delta u_i + {\mathpzc{W}}_i^\mathrm{x}\label{eqn:unified_state_prop}\\
                                 %& \quad 0 \geq H_i^\mathrm{x} \Delta x_i + H_i^\mathrm{u}\Delta u_i + \mathpzc{W}_i^\mathrm{h}.
%\end{align}
%\end{subequations}
\interdisplaylinepenalty=10000
\begin{subequations}
\label{eqn:QP_of_LPV}
    \begin{align}
    \min_{\Delta X,\Delta U} &\quad \sum_{i=0}^{N-1}\frac{1}{2}
     \begin{bmatrix}\Delta x_i \\ \Delta u_i \end{bmatrix}^\top \mathpzc{M}_i\begin{bmatrix}
        \Delta x_i\\
        \Delta u_i
    \end{bmatrix} + \left(M_i\begin{bmatrix}
        \hat{x}_i\\ \hat{u}_i
    \end{bmatrix}\right)^\top \begin{bmatrix}
        \Delta x_i\\ \Delta u_i
    \end{bmatrix} \nonumber \\
    & + \frac{1}{2} \Delta x_N^\top \mathpzc{M}_N \Delta x_N + (M_N \hat{x}_N)^\top \Delta x_N \\
        %\min_{\Delta z} &\quad \frac{1}{2} \Delta z^{\top} M \Delta z+\frac{\partial \xi}{\partial z}(\hat{z}) \Delta z, \\
        %\textrm{s.t.} &\quad \eta(\hat{z})+ \frac{\partial \eta}{\partial z}(\hat{z})\Delta z=0,\\
        %&\quad \zeta(\hat{z})+ \frac{\partial \zeta}{\partial z}(\hat{z})\Delta z\leq0,
        \textrm{s.t.} &  \quad \forall i\in\mathbb{I}_0^{N-1}\\
        & \quad \Delta x_{i+1}=A_i\Delta x_i+ B_i\Delta u_i + \Delta\mathpzc{W}_i^\mathrm{x},\;\label{eqn:state_prop_in_unified}  \\
                   & \quad 0\geq H^\mathrm{x}_i\Delta x_i+H_i^\mathrm{u}\Delta u_i+ \Delta\mathpzc{W}_i^\mathrm{h},\\
                   & \quad 0\geq H^\mathrm{x}_N\Delta x_N+ \Delta\mathpzc{W}_N^\mathrm{h},\\
                  & \quad \Delta x_0=0.
    \end{align}
\end{subequations}
In SQP, $\mathpzc{M}_i$ denotes the approximated Hessian of the Lagrangian corresponding to stage~$i$ (see Appendix~\ref{app:hessian}), while for the LPV-MPC, \mbox{$\mathpzc{M}_i=M_i=\mathrm{diag}(Q,R)$}, for all \mbox{$i\in\mathbb{I}_0^{N-1}$} and \mbox{$\mathpzc{M}_N=M_N=W$}, i.e., it is composed of the weighting matrices of the MPC cost. Note that by employing \emph{Gauss-Newton} (GN) approximation for SQP~\cite[Sec.~3.1]{Gros_Zanon_Quirynen_Bemporad_Diehl_2020}, we retrieve the same block diagonal matrix, as the approximation neglects the dependence of the Lagrangian on the constraints. In both cases, \mbox{$M_i = \mathrm{diag}(Q,R)$}. %The state and input matrices for the mean propagation are the Jacobians for the SQP at $\hat{z}$:
To get a better overview, all the parameters of \eqref{eqn:QP_of_LPV} are collected in \Cref{TABLE:SQP_vs_LPV}. In conclusion, it is important to emphasize that the LPV iterations use the integrated Jacobians as transition matrices to obtain the exact embedding of the nonlinear dynamics, whereas the SQP methods rely on the Jacobians obtained through linearization.

\allowdisplaybreaks

The unified formulation yields the following results.
\begin{proposition}[Equivalence of SQP and LPV-MPC]
\label{prop:sqp-lpvmpc-equivalence}
Consider the LPV-MPC formulation \eqref{eqn:q_linear_MPC} with the FTC-based LPV embedding.  
If the stage-wise anchor points are chosen as the previous solutions, i.e.,
$
\tilde{x}_i = \hat{x}_i,\; i\in\mathbb{I}_0^N,$
$\tilde{u}_i = \hat{u}_i,\; i\in\mathbb{I}_0^{N-1},
$
then the LPV-MPC iteration coincides exactly with the SQP iteration for the original nonlinear MPC problem. Consequently, at convergence, the solution is (locally) optimal.
\end{proposition}

\begin{proof}
When the anchor points are set as $\tilde{x}_i=\hat{x}_i$ and $\tilde{u}_i=\hat{u}_i$, the LPV system and constraint matrices computed by \eqref{eqn:FTC_to_LPV} reduce to the Jacobians of \eqref{eqn:state_progagation} and \eqref{eqn:nlp_constraints}, respectively. %:
%\[
%A_i=A(\hat{\rho}_i) = \left. \frac{\partial f} {\partial x} \right | _ {\substack{\hat{x}_i \\ \hat{u}_i}}, 
%\qquad 
%B_i=B(\hat{\rho}_i) = \left. \frac{\partial f} {\partial u} \right | _ {\substack{\hat{x}_i \\ \hat{u}_i}}.
%\]
%Analogously, the inequality constraint matrices reduce to
%\[
%H_i^\mathrm{x}=H^\mathrm{x}(\hat{\rho}_i) = \left. \frac{\partial h} {\partial x} \right | _ {\substack{\hat{x}_i %\\ \hat{u}_i}}, 
%\quad 
%H_i^\mathrm{u}=H^\mathrm{u}(\hat{\rho}_i) = \left. \frac{\partial h} {\partial u} \right | _ {\substack{\hat{x}_i %\\ \hat{u}_i}}.
%\]
Consequently, both the equality constraints and the inequality constraints in \eqref{eqn:QP_of_LPV} coincide exactly with the first-order Taylor expansions used in SQP.  
Therefore, the LPV--MPC step is equivalent to the SQP subproblem, and the standard SQP convergence results apply~(see Lemma~\ref{lem:sqp-conv}), ensuring local convergence to a KKT point of the original NLP.
\end{proof}

\begin{corollary}
\label{cor:opt_conv_lpv}
Let $X^\star,\;U^\star$ denote a KKT point of \eqref{eqn:NLP}. If 
% the anchor points of LPV-MPC are locally optimal, i.e., 
$\tilde{X}=X^\star,\;\tilde{U}=U^\star$, then, $X^\star,\;U^\star$ is a stationary point of the LPV-MPC algorithm (Alg.~\ref{alg:LPV-MPC}).
\end{corollary}

\begin{proof}
Let $\hat{X}=X^\star,\;\hat{U}=U^\star$, i.e., the solution of the last QP corresponds to the optimal solution. Then, 
since $\hat{X} = \tilde{X}$ and $\hat{U} = \tilde{U}$, 
according to Proposition \ref{prop:sqp-lpvmpc-equivalence}, the next iterate coincides with the SQP solution. However, as the last solution corresponds to a KKT point of \eqref{eqn:NLP}, the SQP step yields $\Delta X=0,\; \Delta U = 0$, i.e., $X^\star,\;U^\star$ is a stationary point of the LPV-MPC algorithm (Alg.~\ref{alg:LPV-MPC}).
\end{proof}

\begin{table}%[htbp]
    \centering
    \vspace{2mm}
    \caption{Comparison of the QPs corresponding to the SQP and LPV-based NMPC solution methods.}
    \label{TABLE:SQP_vs_LPV}
    \begin{tabular}{|c|c |c|}
    \hline
        Parameter & SQP & LPV-MPC \\
        \hline\hline
        $\mathpzc{M}_i$ & $\mathrm{diag}(Q,R)$\footnotemark[2] & $\mathrm{diag}(Q,R)$ \\
        \hline
        ${M}_i$ & $\mathrm{diag}(Q,R)$\footnotemark[2] & $\mathrm{diag}(Q,R)$ \\
        \hline
         $A_i$ & $ \left.\frac{\partial f}{\partial x}\right|_{\substack{\hat{x}_i\\\hat{u}_i}} $ & $\int_0^1 \left. \frac{\partial f}{\partial x}\right|_{\substack{\tilde{x}_i+\lambda (\hat{x}_i-\tilde{x}_i) \\ \tilde{u}_i+ \lambda (\hat{u}_i-\tilde{u}_i)}}\mathrm{d}\lambda$\\
         \hline
         $B_i$ & $\left.\frac{\partial f}{\partial u}\right|_{\substack{\hat{x}_i\\\hat{u}_i}} $ & $\int_0^1 \left. \frac{\partial f}{\partial u}\right|_{\substack{\tilde{x}_i+\lambda (\hat{x}_i-\tilde{x}_i) \\ \tilde{u}_i+ \lambda (\hat{u}_i-\tilde{u}_i)}}\mathrm{d}\lambda$\\
         \hline
         $\Delta\mathpzc{W}_i^\mathrm{x}$ & $  {f}(\hat{x}_i,\hat{u}_i)-\hat{x}_{i+1}$ & $   {f}(\hat{x}_i,\hat{u}_i)-\hat{x}_{i+1}$\\
         \hline
         $H_i^{\{\mathrm{x},\mathrm{u}\}}$ & $\left.\frac{\partial {h}}{\partial \{x,u\}}\right|_{\substack{\hat{x}_i \\ \hat{u}_i}}$ & $\int_0^1 \left. \frac{\partial h}{\partial \{x,u\}}\right|_{\substack{\tilde{x}_i+\lambda (\hat{x}_i-\tilde{x}_i) \\ \tilde{u}_i+ \lambda (\hat{u}_i-\tilde{u}_i)}}\mathrm{d}\lambda$ \\
         \hline
         $\Delta\mathpzc{W}_i^\mathrm{h}$ & ${h}(\hat{x}_i,\hat{u}_i)$ & ${h}(\hat{x}_i,\hat{u}_i)$\\
         \hline
    \end{tabular}
    \vspace{-5mm}
\end{table}
\subsection{Zero-order Approximation}
\label{sec:unified_zoro}
For complex systems, it is often necessary to employ simplifications of the MPC scheme to ensure computational feasibility.
%For example, in the case of robust MPC schemes, the states used for the uncertainty description are usually propagated outside the optimization loop~\cite{Hewing_Kabzan_Zeilinger_2020}. However, fixing these states can result in infeasibilities for the original problem~\cite{Lahr_Zanelli_Carron_Zeilinger_2023}. 
A commonly used approach is to apply a zero-order approximation, where a tailored Jacobian structure allows one component of the state to be computed independently of the remaining variables, enabling it to be propagated outside the optimization problem, while the other components still depend on it within the optimization. This method has been successfully applied for robust~\cite{Zanelli_Frey_Messerer_Diehl_2021} and stochastic \cite{Feng_Cairano_Quirynen_2020,Lahr_Zanelli_Carron_Zeilinger_2023} MPC schemes to eliminate the uncertainty description from the optimization variables. In the following, we derive the zero-order approximation for the unified MPC description of Sec.~\ref{sec:unified_nominal} using the differential formulation.

Let the states be divided as $x^\top=[y^\top\; z^\top]^\top$, where ${y}$ are the states considered as optimization variables and $z$ are the states to be propagated outside the optimization loop. Then the equality constraints corresponding to the state propagation \eqref{eqn:state_prop_in_unified} can be formulated as
\begin{equation}
    \begin{bmatrix}
        \Delta y_{i+1}\\
        \Delta z_{i+1}
    \end{bmatrix}=\begin{bmatrix}
        A^\mathrm{yy}_i & A^\mathrm{yz}_i\\
        A^\mathrm{zy}_i & A^\mathrm{zz}_i
    \end{bmatrix}\begin{bmatrix}
        \Delta y_i \\ \Delta z_i 
    \end{bmatrix}
    +\begin{bmatrix}
        B^\mathrm{y}_i\\
        B^\mathrm{z}_i
    \end{bmatrix}
    \Delta u_i + \begin{bmatrix}
        \Delta{\mathpzc{W}}_i^\mathrm{y}\\
        \Delta{\mathpzc{W}}_i^\mathrm{z}
    \end{bmatrix}.
\end{equation}
In the zero-order method, the following simplifications are made: $A^\mathrm{zy}_i=0$, $B^\mathrm{z}_i=0$. As a result, the evolution of $z$ no longer depends on the optimization variables $\Delta y, \Delta u$ and can be simplified as 
\begin{equation}
\label{eqn:z_zero_order}
    \Delta z_{i+1} = A_i^\mathrm{zz} \Delta z_i + \Delta{\mathpzc{W}}^\mathrm{z}_i,
\end{equation}
where \mbox{$\Delta{\mathpzc{W}}_i^\mathrm{z}=-\hat{z}_{i+1} + f^\mathrm{z}(\hat{y}_i, \hat{z}_i, \hat{u}_i)$} corresponds to the $z$-component of the original dynamics~\eqref{eqn:NL_sys}. As outlined in~\cite{Feng_Cairano_Quirynen_2020}, there are multiple approaches to propagate the dynamics of $z$ in between solver iterations. First, noticing that \mbox{$\Delta z_0=0$}, \eqref{eqn:z_zero_order} can be rolled out to obtain the sequence of auxilary variables. Second,~\cite{Feng_Cairano_Quirynen_2020} also suggests the propagation of $z$ based on the original nonlinear dynamics, i.e.,
\begin{equation}
\label{eqn:z_prop}
    z_{i+1}=f(\hat{y}_i,z_i, \hat{u}_i).
\end{equation}
Note that if \eqref{eqn:z_prop} is linear in the auxiliary variable $z$, the two methods produce identical results. Consequently, since most iterative \mbox{LPV-MPC} approaches addressing uncertainty (e.g.~\cite{Polcz_Peni_Toth_2023}) use this form of auxiliary propagation, they can be interpreted as a zero-order approximation known from the SQP scheme. %Hence, 
% in the proposed unified framework, it can be seen that 
%an equivalence can be established between the SQP and iterative LPV-MPC formulations that implement auxiliary state propagation and use the extended scheduling variable $\rho= [y\;u\;z]^\top$
% uncertainty-based scheduling variables 
%for the underlying MPC solver loop.
%in the proposed unified framework, it can be seen that the approximations introduced
The proposed unified framework thereby allows to identify the correspondence and shows that the approximations introduced
by i) the zero-order method in SQP and ii) the iterative \mbox{LPV-MPC} formulation that uses auxiliary state propagation~\eqref{eqn:z_prop} and an extended scheduling variable $\rho= [y^\top\;u^\top\;z^\top]^\top$ are equivalent.

%For the LPV-MPC scheme, the state evolution of $z$ can be expressed as
%    \begin{align}
%\label{eqn:chi_full_differential}
%    \Delta z_{i+1}+\hat{z}_{i+1} &= A^\mathrm{zy}(\hat{\rho}_i)\Delta y_i+B^\mathrm{z}(\hat{\rho}_i)\Delta u_i+A^\mathrm{zz}(\hat{\rho}_i)\Delta z_i\\ &\phantom{=}+ A^\mathrm{zy}(\hat{\rho}_i)(\hat{y}_i-\tilde{y}_i)+B^\mathrm{z}(\hat{\rho}_i)( \hat{u}_i - \tilde{u}_i)\\&\phantom{=}+A^\mathrm{zz}(\hat{\rho}_i) (\hat{z}_i - \tilde{z}_i)  + V^\mathrm{z}_i.
%\end{align}
%According to the zero-order approximation, we neglect the dependence of $\Delta z_i$ on the optimization variables $\Delta {y}_i$, $\Delta u_i$, to obtain
%\begin{align}
%     z_{i+1} &=  A^\mathrm{zy}(\hat{\rho}_i)(\hat{y}_i-\tilde{y}_i)+B^\mathrm{z}(\hat{\rho}_i)(\hat{u}_i-\tilde{u}_i) \\ &\phantom{=}+A^\mathrm{zz}(\hat{\rho}_i)(z_i - \tilde{z}_i) + V^\mathrm{z}_i, \nonumber
%\end{align}
%which corresponds to the FTC-based embedding of $f^\mathrm{z}(\hat{y}_i,z_i,\hat{u}_i)$ along the scheduling variable \mbox{$\hat{\rho}_i=[\hat{y}_i^\top\; \hat{z}_i^\top\; \hat{u}_i^\top]^\top$}, i.e. forward simulation of the $z$-component along the previous solution $\hat{y}_i, \hat{u}_i$. At convergence, $\hat{\rho}_i$ corresponds to the suboptimal state and input sequence that satisfies the constraints, and, therefore, the original NL dynamics.
\footnotetext[2]{Assuming GN Hessian approximation.}
\footnotetext[3]{\url{https://gitlab.ethz.ch/ics/sqp_lpv_mpc}}
\section{Simulation Study}
\label{sec:simulation_study}
Next, we compare the computational complexity and convergence properties of the SQP and the FTC-based LPV-MPC algorithm through the proposed unified form, where each method results as a particular choice of the involved terms. For this, we have implemented both algorithms in \textsc{acados}~\cite{Verschueren_2022} using the \textsc{l4acados} package~\cite{Lahr_Naf_Wabersich_Frey_Siehl_Carron_Diehl_Zeilinger_2024}. The source code is available online\footnotemark[3]. All simulations are carried out using an M2 MacBook Air with 16 GB RAM.

First, we analyze the convergence properties of both algorithms using  a simplified nonlinear example; then, we employ them with RTI for the control of an autonomous race car.

\subsection{Convergence Properties and Computational Complexity}
%To analyze the convergence properties and the computational complexity of both schemes,
First, we utilize the cart-pendulum system, see, e.g.,~\cite[Eq.~ (23)-(24)]{Guemghar_Srinivasan_Mullhaupt_Bonvin_2002},
where $p$, $\dot{p}$ are the position and velocity of the cart, and $\phi$, $\dot{\phi}$ are the angle and angular velocity of the pendulum, jointly defining the state vector \mbox{$x=[p\; \dot{p}\; \phi\; \dot{\phi}]^\top$}. We aim to steer the system to the equilibrium $x^\mathrm{eq}=0$ from a downward initial position $x[0]=[0\;0\;-\pi\; 0]^\top$. We consider box state and input constraints in the form \mbox{$x_i \in [-5,5]\times[-5,5]\times[-2\pi,2\pi]\times[-10,10]$, $u_i \in [-4,4]$}. To discretize the cart-pendulum system, we utilize a \emph{fourth-order Runge-Kutta} (RK4) numerical integration method with $T_\mathrm{s}=0.01$s sampling time. The prediction horizon is $N=20$. Furthermore, for the LPV-MPC, we use a rectangular numerical integration scheme with $n_\mathrm{int}=20$ stages to compute the integral of the FTC-based embedding~\eqref{eqn:FTC_to_LPV}. To formulate the MPC cost \eqref{eqn:NL_MPC_cost}, we use \mbox{$Q=\mathrm{diag}(100,1,100,1)$}, $R=10$.

% During the simulations, 
We compare five different LPV-MPC algorithms with the SQP scheme: (1) constant anchor points at the origin $\tilde{x}_i=0$, $\tilde{u}_i=0$; (2) constant non-zero\footnotemark[4] anchor points $\tilde{x}_i=c^\mathrm{x}$, $\tilde{u}_i=c^\mathrm{u}$; (3) last input and measured state as anchor points \mbox{$\tilde{x}_i=x[k]$, $\tilde{u}_i=u[k-1]$}\footnotemark[5], with $u[-1] \doteq 0$; (4) last optimizer sequence as anchor points $\tilde{x}_i=\hat{x}_i$, $\tilde{u}_i=\hat{u}_i$; (5) the idealized setting of optimal state and input sequence as anchor points $\tilde{x}_i=x^\star_i$, $\tilde{u}_i=u^\star_i$. For a fair comparison, both algorithms use the same initialization $\hat{x}_i=x[0], \hat{u}_i=0$.% and termination criteria \eqref{eqn:lpv_convergence_criteria}, with $\epsilon=10^{-6}$. Note that this differs from the usual SQP termination criteria, which rely on KKT residuals.
\footnotetext[4]{The $c^\mathrm{x}, c^\mathrm{u}$ values are picked randomly from the feasible set, then kept constant during the simulations.}
\footnotetext[5]{Note that this is different from gain-scheduled MPC, where the anchor points are 0 and the scheduling trajectory is set to be constant and equal to the previous state- and input-induced values.}

In Fig. \ref{fig:residuals}, we evaluate the NLP residuals of the original NMPC problem at the first time step by performing a fixed number of 10 iterations. As shown, the LPV-MPC algorithm generally converges to a suboptimal solution of the original NLP, according to Lemma~\ref{lem:lpvmpc-suboptimal}. When the last optimizer sequence is used as the anchor sequence, we retain the SQP algorithm and its convergence properties, verifying Proposition~\ref{prop:sqp-lpvmpc-equivalence}. Furthermore, if a KKT point is used as anchor points, the LPV-MPC maintains this stationary point (Corollary~\ref{cor:opt_conv_lpv}). %notably, faster than the SQP scheme.

\begin{figure}
\vspace{1mm}
\centering\includegraphics[width=\linewidth]{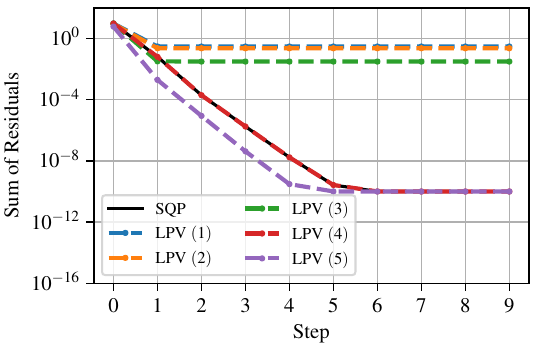}
    \vspace{-3mm}
    \caption{KKT residuals of the SQP and LPV-MPC iterations for a single OCP for the cart-pendulum system.}
    \label{fig:residuals}
    \vspace{-4mm}
\end{figure}

In Fig. \ref{fig:num_iters}, the number of solver iterations required to converge is shown along an 80-step rollout of the closed-loop system. For a fair comparison, both algorithms use the same LPV-MPC termination criteria~\eqref{eqn:lpv_convergence_criteria}, with $\epsilon=10^{-6}$. Note that this differs from the usual SQP termination criterion based on KKT residuals. As shown in Fig.~\ref{fig:residuals}, both methods can yield solutions with a varying degree of optimality and iteration number.

Furthermore, \Cref{tab:pend_sim_times} details the computational costs associated with the preparation (construction of the QP) and the feedback (solving the QP) phases per iteration, averaged over the whole rollout. For this example, LPV-MPC approaches require fewer iterations to converge at the expense of suboptimality. However, while the SQP algorithm generally needs more iterations to converge, computing the Jacobian is cheaper than evaluating the integral \eqref{eqn:FTC_to_LPV}, keeping the total solution time comparable. Still, for large-dimensional systems where the reduction in QP iterations dominates the additional cost of the integration scheme~\eqref{eqn:FTC_to_LPV}, the LPV-MPC algorithm can be advantageous, especially since the integration can be easily parallelized and further tuned through advanced quadrature schemes or fewer integration stages.
\begin{figure}
    \centering
    \includegraphics[width=\linewidth]{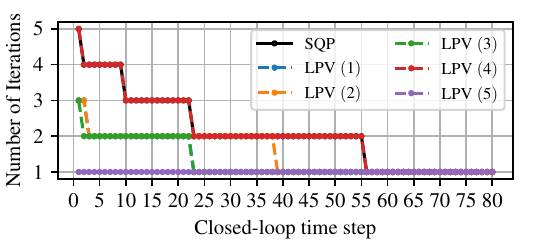}
    \vspace{-5mm}
    \caption{Number of iterations required to converge at each closed-loop step for the cart-pendulum system. Note that LPV (4) overlays SQP, verifying Proposition~\ref{prop:sqp-lpvmpc-equivalence}.}
    \label{fig:num_iters}
    \vspace{-3mm}
\end{figure}

\begin{table}
\vspace{2mm}
    \centering
    \caption{Computational times and number of iterations required to converge, averaged through the rollout for the cart-pendulum simulation. Preparation time $T_\mathrm{prep}$, feedback time $T_\mathrm{fb}$ are given in milliseconds, while $n_\mathrm{it}$ denotes the number of iterations.}
    \label{tab:pend_sim_times}
    \begin{tabular}{|c|c||c|c|c|}
    \hline
        Method & $(\tilde{x}_i.\;\tilde{u}_i)$ & $T_\mathrm{prep}$  & $T_\mathrm{fb}$ & $n_\mathrm{it}$\\
        \hline
        \hline
        SQP & -- & 0.53 & 0.46  & 2.09\\
        \hline
        LPV (1) &($0,\;0$) & 0.78 & 0.46  & 1.3\\
        \hline
        LPV (2)&($c^\mathrm{x},\; c^\mathrm{u}$) & 0.79 & 0.46  & 1.5\\
        \hline
        LPV (3)&($x[k],\; u[k-1]$) & 0.76 & 0.45 & 1.28\\
        \hline
        LPV (4)&($\hat{x}_i,\; \hat{u}_i$) & 0.75 & 0.46  & 2.09
        \\
        \hline
        LPV (5)& (${x}^\star_i,\; {u}^\star_i$) & 0.75 & 0.45  & 1\\
        \hline
    \end{tabular}
    \vspace{-5mm}
    \end{table}

\subsection{Autonomous Racing}
This section applies the LPV-MPC algorithm for autonomous racing with MPCC~\cite{Hewing_Kabzan_Zeilinger_2020, Lahr_Naf_Wabersich_Frey_Siehl_Carron_Diehl_Zeilinger_2024}. We first outline the \emph{autonomous ground vehicle}~(AGV) model and the resulting LPV-MPC formulation.
\subsubsection{Vehicle Model}
\label{sec:vehicle_model}
We use a dynamic single-track model~\cite[Eq.~(5)]{Carron_Bodmer_Vogel_Zurbrugg_Helm_Rickenbach_Muntwiler_Sieber_Zeilinger_2023} to describe the motion dynamics, where the state vector %\mbox{$x=\begin{bmatrix}
%    p_\mathrm{x}& p_\mathrm{y}& \psi& v_\mathrm{x}& v_\mathrm{y}& \omega& T & \delta& \theta\end{bmatrix}^\top$}
$x=[p_\mathrm{x}\; p_\mathrm{y}\; \psi\; v_\mathrm{x}\; v_\mathrm{y}\; \omega\; T\; \delta\; \theta]^\top$
comprises the 2D position of the vehicle $(p_\mathrm{x}, p_\mathrm{y})$, heading angle $\psi$ with respect to the global $x$-axis, longitudinal and lateral velocities $(v_\mathrm{x}, v_\mathrm{y})$, and yaw rate $\omega$ in the body-fixed frame. Additionally $T$ is the applied motor torque and $\delta$ is the steering angle. Lastly, $\theta$ is the progress along the track. Overall the model is obtained by combining single track dynamics, a Pacejka tire model and integrators for the torque and steering dynamics (modeling the low-level torque and steering controllers) and the progress variable. Consequently, the control input is $u = [\dot{ T}\;  \dot{\delta}\;\dot{\theta}]^\top$. To obtain the DT dynamics,
% for the MPCC formulation, 
we discretize the model by RK4 to obtain 
\begin{equation}\label{eqn:MPCC_state_prop}
x[k+1]=f_\mathrm{C}(x[k],u[k]).
\end{equation}
\subsubsection{LPV-MPCC formulation}
\label{sec:LPV_MPCC}
The key idea of the MPCC algorithm is to maximize the progress along a predefined reference path while minimizing the deviation from it and respecting the constraints imposed by the boundaries of the reference track. %Formally, it can be expressed by the NLP:
%\begin{subequations}
%\label{eqn:MPCC_NLP}
%\begin{align}
%    \min_{U} & \quad\sum_{i=0}^{N-1} \norm{e(x_i)}_Q^2 -q\theta(x_i) + \norm{u_i}_R^2\label{eqn:MPCC_cost}\\
 %   \textrm{s.t.} & \quad \forall i\in\mathbb{I}_0^{N-1},\\
 %   & \quad x_{i+1} = f_\mathrm{MPCC}(x_i, u_i),\; \\
 %   & \quad 0\geq h(x_i,u_i),\;\theta(x_i)\geq0,\label{eqn:NL_MPCC_const}\\
 %   & \quad x_0 = x[k],
%\end{align}
%\end{subequations}
%where \mbox{$e(x)=[e_\mathrm{l}(x)\; e_\mathrm{c}(x)]^\top$} denotes the contouring and the lag error~\cite{Hewing_Kabzan_Zeilinger_2020} from the predefined reference. Furthermore, the progress of the vehicle along the reference is $\theta(x_i)$, while $h(x_i,u_i)$ encodes the constraints. Finally, $Q$, $q$, and $R$ are the weights of the cost and $N$ is the prediction horizon length. 
To embed the MPCC formulation into the LPV-MPC framework, we define the output equation and the track constraints as
\begin{align}
    \mathpzc{y}_i &= c({x}_i), \label{eqn:NL_MPCC_output}\\
    0&\geq h(x_i, u_i),\label{eqn:NL_MPCC_const}
   \end{align}
where \mbox{$\mathpzc{y}_i=[e_\mathrm{l}\; e_\mathrm{c}\; \theta\; 1]^\top$} contains the contouring and lag errors~\cite[Sec.~IV.B]{Hewing_Liniger_Zeilinger_2018}.
Then, using the FTC-based embedding for \eqref{eqn:MPCC_state_prop}--\eqref{eqn:NL_MPCC_const} with scheduling variable \mbox{$\rho_i=[{x}_i^\top\; {u}_i^\top]^\top$}, the OCP of the LPV-MPCC algorithm can be expressed as
\begin{subequations}
\begin{align}
    \min_{{X},{U}} & \quad \quad\sum_{i=0}^N \norm{\mathpzc{y}_i}_{\tilde{Q}}^2 + \norm{{u}_i^\top}_{\tilde{R}}^2  \\
    \textrm{s.t.} & \quad \forall i\in\mathbb{I}_0^{N-1}\\
     &\quad {x}_{i+1}=A(\rho_i){x}_i +B(\rho_i){u}_i + \mathpzc{W}^\mathrm{x}_i,\\
     & \quad \mathpzc{y_i} = C(\rho_i){x}_i + \mathpzc{W}_i^\mathrm{y},\\
     & \quad H^\mathrm{x}(\rho_i){x}_i + H^\mathrm{u}(\rho_i){u}_i + \mathpzc{W}^\mathrm{h}_i \leq 0,\\
     &\quad {x}_0={x}[k],
\end{align}
\end{subequations}
where $\tilde{Q}$ is the weighting matrix of the contouring and lag error and $\tilde{R}$ is the input weighting matrix as outlined in \cite{Lahr_Naf_Wabersich_Frey_Siehl_Carron_Diehl_Zeilinger_2024}.

\subsubsection{Simulation Results}
In the following simulations, we compare how the LPV-MPC and the SQP algorithms perform in an RTI scheme. The simulations are performed using the simulator module of CRS~\cite{Carron_Bodmer_Vogel_Zurbrugg_Helm_Rickenbach_Muntwiler_Sieber_Zeilinger_2023}, which uses the dynamic model of a 1/28 scale autonomous electrical car.

During the simulation experiments, we execute multiple laps around a test track with the controller and compare the average KKT residuals and the residual reduction after each iteration, which is computed as the ratio between the residuals before and after ($\alpha_{\mathrm{r}, \mathrm{avg}} = \sum_0^{N_\mathrm{sim}}\frac{r_{k}}{r_{k+1}}/N_\mathrm{sim}$) each QP solution step. Furthermore, we compare the average preparation and feedback times of the RTI iterations. Note that, since we do not have access to the optimal solution, we omit variant (5) from this study. 

As shown in \Cref{tab:MPCC_sim_results}, the SQP method achieves shorter preparation times than the iterative LPV-MPC, because the Jacobians are evaluated only once per step along the prediction horizon, whereas the LPV-MPC requires multiple evaluations for numerical integration. Given that the resulting QPs have similar structures, the feedback times are comparable. However, LPV-MPC generally exhibits a larger average reduction in NLP residuals and a smaller average residual value. This indicates that the LPV-MPC algorithm may tend to operate closer to optimality in the RTI framework despite its suboptimality at convergence, due to a more effective global embedding. Lastly, note that with a suitable selection of anchor points (4), the SQP and LPV-MPC iterations become equivalent.

\begin{table}
    \centering
    \vspace{2mm}
    \caption{Comparison of the NMPC schemes for autonomous racing simulations. $T_\mathrm{prep}$ and $T_\mathrm{fb}$ are in ms.}
    \label{tab:MPCC_sim_results}
    \begin{tabular}{|c|c||c|c|c|c|}
    \hline
        Alg. & ($\tilde{x}_i,\; \tilde{u}_i$) & $T_\mathrm{prep}$  & $T_\mathrm{fb}$ & $\alpha_{\mathrm{r},\mathrm{avg}}$ & $r_\mathrm{avg}$\\
        \hline
        \hline
        SQP & -- & 13.7 & 9.21 & 1.30 & 3.45\\
        \hline
        LPV (1) &($0,\; 0$) & 17.6 & 9.4 & 1.54 & 3.23\\
        \hline
        LPV (2)& ($c^\mathrm{x},\; c^\mathrm{u}$) & 17.6 & 9.3 & 1.64 & 3.33\\
        \hline
        LPV (3)&($x[k],\;u[k-1]$)& 17.6 & 9.4 & 1.67 & 3.19\\
        \hline
        LPV (4)&($\hat{x}_i,\; \hat{u}_i$) & 17.6 & 9.4 & 1.30 & 3.45 \\
        \hline
    \end{tabular}
    \vspace{-3mm}
\end{table}
\section{Experiments}
\label{sec:experiment}
We perform real-world experiments applying learning-based MPCC on the small-scale vehicle platform, allowing us to evaluate the zero-order approximation scheme\footnotemark[6]. The setup is based on CRS~\cite{Carron_Bodmer_Vogel_Zurbrugg_Helm_Rickenbach_Muntwiler_Sieber_Zeilinger_2023}, which employs custom 1/28-scale electric cars and a Qualisys motion capture system. As in simulation, the controller is formulated as an MPCC (Sect.~\ref{sec:LPV_MPCC}), but in experiments, we augment the nominal model with a GP to learn the residual dynamics inherently present when working with real hardware. Then, we utilize the stochastic GP-MPC scheme of \cite{Lahr_Zanelli_Carron_Zeilinger_2023}.
\footnotetext[6]{Experimental data available at \doi{10.3929/ethz-c-000797782}.}

Formally, we consider $x[k+1]=f_\mathrm{C}(x[k],u[k])+B_\mathrm{g}g(x[k],u[k])+w[k]$, where $w[k]$ is the process noise, $g: \mathbb{R}^{n_\mathrm{x}}\times\mathbb{R}^{n_\mathrm{u}}\rightarrow \mathbb{R}^{n_\mathrm{g}}$ is the unknown residual dynamics and $B_\mathrm{g}\in\mathbb{R}^{n_\mathrm{x}\times n_\mathrm{g}}$ is a full column rank matrix, characterizing that $g$ only affects a subspace of the full state space. As most significant modeling errors usually appear in the tire and drivetrain parameter estimates~\cite{Lahr_Naf_Wabersich_Frey_Siehl_Carron_Diehl_Zeilinger_2024,Floch_Peni_Toth_2024}, we define %$B_\mathrm{g}\doteq \begin{bmatrix}
    %0_{3\times 3} & I_{3\times 3 }& 0_{3 \times 3}\end{bmatrix}$
\mbox{$B_\mathrm{g}\doteq [0_{3\times 3}\; I_{3\times 3 }\; 0_{3 \times 3}]^\top$}
and estimate $g$ with GPs, i.e.,
%    \begin{equation}
$        g\sim \mathcal{GP}(\mu_\mathrm{g},\Sigma_\mathrm{g}),$
%    \end{equation}
where $g: \mathbb{R}^{n_\mathrm{x}}\times\mathbb{R}^{n_\mathrm{u}}\rightarrow \mathbb{R}^{n_\mathrm{g}}$ is the posterior mean and \mbox{$\Sigma_\mathrm{g}: \mathbb{R}^{n_\mathrm{x}}\times\mathbb{R}^{n_\mathrm{u}}\rightarrow \mathbb{R}^{n_\mathrm{g} \times n_\mathrm{g}}$} is the posterior variance.
As the computational demand of the naive GP-MPC scales quadratically with the number of system states, we utilize the zero-order approximation (Sec.~\ref{sec:unified_zoro}) for the propagation of covariances. The formulation and implementation of the GP-MPC %for autonomous racing
are based on~\cite{Lahr_Zanelli_Carron_Zeilinger_2023, Lahr_Naf_Wabersich_Frey_Siehl_Carron_Diehl_Zeilinger_2024}.

%For GP estimates, we utilize the \textsc{GPytorch}  \cite{Gardner_Pleiss_Bindel_Weinberger_Kilian_Wilson_2018}
The GP implementation is based on \textsc{GPytorch} and uses $D=50$ datapoints collected and updated online according to~\cite[Sec.~IV.D.2]{Lahr_Naf_Wabersich_Frey_Siehl_Carron_Diehl_Zeilinger_2024}. %For real-time feasibility, 
The controller is run in RTI-mode at 30 Hz, with $N=40$ prediction horizon. To compute the integrals \eqref{eqn:FTC_to_LPV}, we use a Gauss-Legendre scheme with $n^\mathrm{int}_\mathrm{nom}=40$ FTC integration stages for the nominal model and $n_\mathrm{GP}^\mathrm{int}=20$ for the GPs. We compare four schemes: nominal (1) SQP and (2) LPV using the measured state and the last input as anchor points, (3) zero-order GP-SQP and (4) zero-order \mbox{GP-LPV}.

As shown in Fig.~\ref{fig:hw_experiments} and Table~\ref{Table:hw_experiments}, the nominal controllers guide the car around the track but fail to follow the optimal raceline. During testing, manually tightened track constraints ($d_\mathrm{t}=0.02$~m) are needed in this case to prevent collisions, whereas the GP-MPCC schemes operated safely without such adjustments, as indicated by the Safe column in Table~\ref{Table:hw_experiments}. In terms of computation, SQP is faster, as it only evaluates Jacobians~\eqref{eqn:SQP_AB_jacobian} $N$ times, while the LPV method requires $N(n^\mathrm{int}_\mathrm{nom}+n_\mathrm{GP}^\mathrm{int})$ evaluations for numerical integration. Although parallelization mitigates this, SQP remains more efficient in RTI schemes. In the learning-based setting, LPV-MPC achieves a lower average cost through improved model approximations, whereas the nominal case incurs higher costs due to a more significant model mismatch.

\begin{figure}
    \centering
    \includegraphics[width=\linewidth]{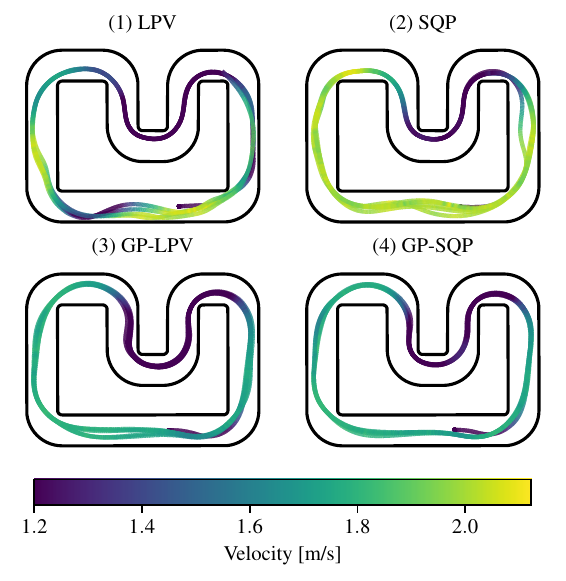}
    \vspace{-8mm}
    \caption{Miniature car racing hardware experiments with the MPCC implementations. A video of the experiments is available at \url{https://youtu.be/-1toeTJSsgg}.}
\label{fig:hw_experiments}
    \vspace{-6mm}
\end{figure}

\begin{table}
\centering
\vspace{2mm}
\caption{Comparison of the MPCC implementations in real hardware experiments. Times are in milliseconds.}
\label{Table:hw_experiments}
\begin{tabular}{|l||c|c|c|c|c|}
\hline
Alg. & $T_\mathrm{sol}$ & $T_\mathrm{prep}$ & $T_\mathrm{fb}$ & Cost & Safe\\
\hline
\hline
LPV & 20.54 & 17.60 & 2.93 & 5.53 & \\
\hline
SQP & 5.13 & 2.17 & 2.95 & 4.37 & \\
\hline

GP-LPV & 31.31 & 25.83 & 5.45 & 6.88 & \checkmark \\
\hline
GP-SQP & 23.72 & 18.63 & 5.07 & 7.46 & \checkmark \\
\hline
\end{tabular}
\vspace{-3mm}
\end{table}
\section{Conclusion}
\label{sec:conclusion}
This paper presented a unified NMPC solution framework that integrates SQP and LPV-MPC as specific subcases. We showed that by the appropriate choice of the sensitivity matrices, both algorithms can be implemented within a common framework, for which we provide an open-source implementation. In particular, we demonstrated that the FTC embedding for LPV-MPC recovers SQP under a specific choice of anchor points. Furthermore, we integrated the zero-order Jacobian approximation into the unified framework and showed its connection to LPV scheduling variables. Finally, in simulations, we highlighted their convergence properties and computational complexity and deployed the algorithms in real-world autonomous racing experiments. %A video summarizing the paper and the hardware experiments is available at \url{https://youtu.be/XXXXXXXXXXXXXXXXXXXXXX}.

%\input{uncertain_LPV}
%\input{GP_MPC_problem}
%\input{SQP}
%\input{LPV-MPC}
%\input{unifying_stochastic_MPC}
%\input{MPCC}
%\input{simulation_study}
%\input{experiment}
% \fi

\bibliographystyle{IEEEtran}
\bibliography{IEEEabrv,references}
\appendix
\subsection{Hessian Approximations in SQP}
\label{app:hessian}
The Lagrangian of the NMPC~\eqref{eqn:NLP} can be expressed as
\begin{multline}
    \mathcal{L}(X,U,\Theta, \mathcal{N})= \sum_{i=0}^{N-1}\Big(l_i(x_i,u_i)+\vartheta_{i+1}^\top(f(x_i,u_i)-x_{i+1}) \\+ \mu_i^\top h(x_i,u_i)\Big)+l_N(x_N) + \mu_N^\top h(x_N, u_N) + \vartheta_0(x_0-x[k]),
\end{multline}
where $\Theta = [\vartheta_0^\top \dots \vartheta_{N}^\top]^\top$ and $\mathcal{N}=[\mu_0^\top\dots\mu_N^\top]^\top$  are the Lagrange multipliers, respectively. Using the exact Hessian of the Lagrangian
\begin{align}
\mathpzc{M}_0 & =  \nabla_{(x_0,u_0)}^2\mathcal{L} = \nabla^2 l_0 + \vartheta_0^\top \mathrm{diag(I_{n_\mathrm{x}\times n_\mathrm{x}}, 0_{n_\mathrm{u}\times n_\mathrm{u}})}\nonumber  &\phantom{=} \nonumber\\ & \phantom{=}+ \mu_0^\top \nabla^2h(x_0,u_0),\\
\mathpzc{M}_i & = \nabla_{(x_i,u_i)}^2\mathcal{L} = \nabla^2 l_i + \vartheta_i^\top\nabla^2 f(x_i,u_i)\nonumber \\ &\phantom{=} + \mu_i^\top \nabla^2h(x_i,u_i), \quad \forall i\in\mathbb{I}_1^{N-1}, \\
\mathpzc{M}_N &= \nabla_{(x_N,u_N)}^2\mathcal{L} = \nabla^2 l_N + \mu_N^\top \nabla^2h(x_N,u_N).
\end{align}
Under the GN approximation the Hessian of the Lagrangian for quadratic cost~\eqref{eqn:quadratic_stage_and_terminal_cost} is \mbox{$ \nabla^2_{(x_i,u_i)}\mathcal{L}\approx \nabla^2_{(x_i,u_i)}l_i(x_i,u_i),$}
i.e., the constraint curvature terms are neglected. Therefore,
\begin{align} \mathpzc{M}_i &=  \nabla_{(x_i,u_i)}^2\mathcal{L} = \mathrm{diag} (Q, R), \quad i\in\mathbb{I}_0^{N-1}\\ \mathpzc{M}_N & =\nabla_{(x_N,u_N)}^2\mathcal{L}= W.
\end{align}
\end{document}